\documentclass[a4paper,UKenglish,cleveref, autoref, thm-restate]{lipics-v2021}

\pdfoutput=1 
\hideLIPIcs  

\title{Fitting Tree Metrics with Minimum Disagreements} 

\author{Evangelos Kipouridis}{Saarland University, Germany \and Max Planck Institute for
Informatics, Germany}{kipouridis@cs.uni-saarland.de}{https://orcid.org/0000-0002-5830-5830}{}

\authorrunning{E. Kipouridis} 

\Copyright{Evangelos Kipouridis} 



\category{} 

\relatedversion{} 

\funding{This work is part of the project TIPEA that has received funding from the European Research Council (ERC)
under the European Unions Horizon 2020 research and innovation programme (grant agreement No. 850979).}

\nolinenumbers 

\EventEditors{Inge Li G{\o}rtz, Martin Farach-Colton, Simon J. Puglisi, and Grzegorz Herman}
\EventNoEds{4}
\EventLongTitle{31st Annual European Symposium on Algorithms (ESA 2023)}
\EventShortTitle{ESA 2023}
\EventAcronym{ESA}
\EventYear{2023}
\EventDate{September 4--6, 2023}
\EventLocation{Amsterdam, the Netherlands}
\EventLogo{}
\SeriesVolume{274}
\ArticleNo{51}

\begin{document}

\maketitle

\begin{abstract}
In the $L_0$ Fitting Tree Metrics problem, we are given all pairwise distances among the elements of a set $V$ and our output is a tree metric on $V$.
The goal is to minimize the number of pairwise distance disagreements between the input and the output.
We provide an $O(1)$ approximation for $L_0$ Fitting Tree Metrics, which is asymptotically optimal as the problem is APX-Hard.

For $p\ge 1$, solutions to the related $L_p$ Fitting Tree Metrics have typically used a reduction to $L_p$ Fitting Constrained Ultrametrics.
Even though in FOCS '22 Cohen-Addad et al. solved $L_0$ Fitting (unconstrained) Ultrametrics within a constant approximation factor, their results did not extend to tree metrics.

We identify two possible reasons, and provide simple techniques to circumvent them.
Our framework does not modify the algorithm from Cohen-Addad et al. It rather extends any $\rho$ approximation for $L_0$ Fitting Ultrametrics to a $6\rho$ approximation for $L_0$ Fitting Tree Metrics in a blackbox fashion. 
\end{abstract}

\section{Introduction}
Trees are used by many disciplines to describe relationships between entities. For example, in biology, the universal tree of life describes evolutionary distances between organisms. In fact, trees are relevant for any historical science studying an evolutionary branching process (e.g. historical linguistics and sociocultural evolution).

In these cases, we are guaranteed that the underlying truth can be described by a tree.
This underlying tree may even have a special structure.
For example in machine learning and data analysis (see e.g.: \cite{ultrametricsMotivation}) it may be an ultrametric, that is a rooted tree with all leaves being at the same depth.
In any case, our access to this tree is usually only through estimations of pairwise distances.
A natural task is thus the reconstruction of the tree, given (noisy) measurements of pairwise distances.

As the noisy measurements may not describe a tree, we are interested in finding the ``closest'' tree to the input.
In this work we study the problem of minimizing the number of pairwise distance disagreements between the measurements and the output tree.
As noted in \cite{vincent}, this objective has a practical relevance; often the distances are obtained by different (human) classifiers. It is expected that most will do a good job, but if an error occurs, it may be by a large amount.

Other objectives have also been studied, e.g. minimizing the total error \cite{charikar, debarati, mcgregor}, or minimizing the maximum error \cite{agarwala}. In order to formally introduce a class of problems that captures all aforementioned objectives we first make some definitions.

\subsection{Problem Definitions}
Given a set $V$, we denote by $\binom{V}{2}$ the set of all (unordered) pairs of disjoint elements from set $V$.
We use the term distance matrix to refer to a function from $\binom{V}{2}$ to the non-negative reals.
Let $D$ be a distance matrix.
We slightly abuse notation and say that for any $u\in V$, $D(u,u)=0$.
For $p \ge 1$, we say that $\|D\|_p = \sqrt[p]{\sum_{\{u,v\} \in \binom{V}{2}}|D(u,v)|^p}$ is the $L_p$ norm of $D$.
We extend the notation for $p=0$.
In this case $\|D\|_0$ denotes the number of pairs $\{u,v\}$ such that $D(u,v)\ne 0$.
We even say $\|D\|_0$ is the $L_0$ norm of $D$, despite $L_0$ not being a norm.
For ease of notation, we use $0^0=0$, so that $x^0=0$ if $x=0$, and $1$ otherwise.
As in \cite{vincent} (and implicitly in \cite{agarwala}), we allow tree metrics and ultrametrics to have distances equal to $0$.

\begin{definition}
In the $L_p$ Fitting Tree (Ultra) Metrics problem, we are given as input a set $V$ 
and a distance matrix $D$.

The output is a tree metric (or ultrametric) $T$ that spans $V$ and fits $D$ in the sense of minimizing the $L_p$-norm
\begin{equation*}
\|T-D\|_p=\sqrt[p]{\sum_{\{u,v\}\in \binom{V}{2}} |T(u,v)-D(u,v)|^p}
\end{equation*}
\end{definition}

We also define a similar problem, $L_p$ Fitting Constrained Ultrametrics.
It was initially defined in \cite{agarwala}.
The authors proved that, for $p\ge 1$, a $\rho$ approximation for $L_p$ Fitting Constrained Ultrametrics translates to a $3\rho$ approximation for $L_p$ Fitting Tree Metrics.

\begin{definition}
In the $L_p$ Fitting Constrained Ultrametrics problem, we are given as input a set $V$,
a distance matrix $D$,
a distinguished element $\alpha \in V$,
a positive number $h$
and a positive number $l_u$ for each $u\in V$.
In particular it holds that $l_\alpha = h$.

The output is an ultrametric $U$ that spans $V$.
It shall also hold that
\begin{equation*}
\max\{l_u,l_v\} \le U(u,v) \le h \quad\quad \forall\{u,v\}\in \binom{V}{2}
\end{equation*}
$U$ shall fit $D$ in the sense of minimizing the $L_p$-norm
\begin{equation*}
\|U-D\|_p=\sqrt[p]{\sum_{\{u,v\}\in \binom{V}{2}} |U(u,v)-D(u,v)|^p}
\end{equation*}
\end{definition}

\subsection{Previous work}
When the input is a tree metric, a corresponding tree can be found in $O(|V|^2)$ time (linear in the input size) \cite{exact}. As this is usually not the case, research focused on $L_p$ Fitting Tree Metrics.

The first $L_p$ Fitting Tree Metrics problem solved within an asymptotically optimal approximation factor is the $L_\infty$ Fitting Tree Metrics problem, by Agarwala et al. \cite{agarwala}.
In order to solve it, the authors give a reduction to the $L_\infty$ Fitting Constrained Ultrametrics problem which increases the approximation by a factor $3$. They then use the exact solution of this problem from \cite{robust}.
In the same paper, they also show how to extend this reduction for any $L_p$ norm, $p\ge 1$.

This reduction turned out to be an essential tool for tackling $L_p$ Fitting Tree Metrics.
Harp, Kannan and McGregor \cite{mcgregor} developed an $O(\min\{n, k \log{n}\}^{1/p})$ approximation factor for $L_p$ Fitting Ultrametrics, $p\ge 1$, where $k$ is the number of distinct distances in the input.
Using the reduction from \cite{agarwala}, they extend their result to the $L_p$ Fitting Tree Metrics case\footnote{The authors erroneously claim that they get the same approximation for the closest tree metric problem. However, the known reduction may create $\omega(k)$ distinct distances. We believe that the dependence in $k$ is polynomial, which makes the approximation worse, but still non-trivial.}.
Similarly, Ailon and Charikar \cite{charikar} get an $O(((\log{n})(\log\log{n}))^{1/p})$ approximation for the ultrametrics case, which they then extend to the tree metrics case using the well established reduction.
Finally, Cohen-Addad et al. \cite{debarati} achieve an asymptotically optimal $O(1)$ approximation factor for $L_1$ Fitting Tree Metrics, again using an asymptotically optimal $O(1)$ approximation factor for the ultrametrics case.

In FOCS '22 Cohen-Addad et al. \cite{vincent} solved the $L_0$ Fitting Ultrametrics problem within an asymptotically optimal $O(1)$ approximation factor.
However, their result was not extended to the $L_0$ Fitting Tree Metrics problem.
We identify two possible reasons for that:
\begin{itemize}
    \item Most importantly, the reduction from \cite{agarwala} does not work for $L_0$. The reason is that a crucial step of it uses the convexity of all $L_p$-norms, $p\ge 1$. $L_0$ however is not convex (and in fact is not a norm).
    \item Even if the reduction worked for $L_0$, the algorithm for $L_0$ Fitting Ultrametrics should be extended to the $L_0$ Fitting Constrained Ultrametrics problem.
\end{itemize}

\subsection{Our results}
In this work we show how any $\rho$ approximation for $L_0$ Fitting Ultrametrics can be extended to a $6\rho$ approximation 
for $L_0$ Fitting Tree Metrics.

In particular, we extend the reduction from \cite{agarwala} to the $L_0$ case, despite $L_0$ not being convex.
We do so by avoiding the averaging argument from \cite{agarwala} which required convexity, and was necessary to prove the existence of a node with certain properties.
Our argument is of course only valid for $L_0$.

Furthermore, we show how one can use any algorithm for $L_0$ Fitting Ultrametrics to solve $L_0$ Fitting Constrained Ultrametrics, in a blackbox manner.
In contrast \cite{agarwala, charikar, mcgregor} all needed to apply ad-hoc modifications to their $L_p$ Fitting Ultrametrics algorithms to also solve $L_p$ Fitting Constrained Ultrametrics.

An immediate corollary of these two results is that the ultrametrics algorithm from \cite{vincent} can be used to get an asymptotically optimal $O(1)$ approximation factor for $L_0$ Fitting Tree Metrics.
Even though this constant is large, any improved approximation factor for $L_0$ Fitting Ultrametrics would immediately yield an improved approximation for $L_0$ Fitting Tree Metrics, using this framework.

Finally, we prove that $L_0$ Fitting Tree Metrics is APX-Hard.

It is interesting to notice that apart from avoiding the averaging argument from \cite{agarwala}, the rest follow existing techniques.
However, due to the special structure of $L_0$, we significantly simplify them.

\section{From tree metrics to ultrametrics}
In this section we prove the following result:

\begin{restatable}{theorem}{thmMain}
\label{thm:main}
A factor $\rho \ge 1$ approximation for $L_0$ Fitting Ultrametrics implies a factor $6 \rho$ approximation for $L_0$ Fitting Tree Metrics.
\end{restatable}

Let $D$ be a distance matrix, $\alpha\in V$ be a distinguished element and $T$ be a tree spanning $V$.
In more details, there exists a function mapping elements from $V$ to nodes in $T$.
If element $u\in V$ is mapped to node $u'\in T$, we say that $u$ is associated with $u'$.
We even say ``node $u$'' to refer to the node associated with $u$.
We note that $T$ may also have auxiliary nodes, without any element from $V$ being mapped to them.

We say $T$ is an $\alpha$-restricted tree if the distance from $\alpha$ to any other element $u$ is the same both in $T$ and in $D$.

Given a tree $T$ we can obtain an $\alpha$-restricted tree $T^{/\alpha}$ by modifying $T$ as in Figure~\ref{fig:aRestricted}. We say that $T^{/\alpha}$ is the $\alpha$-restricted tree of $T$.

\begin{figure}[ht]
    \centering
    \includegraphics[width=320pt, height=150pt]{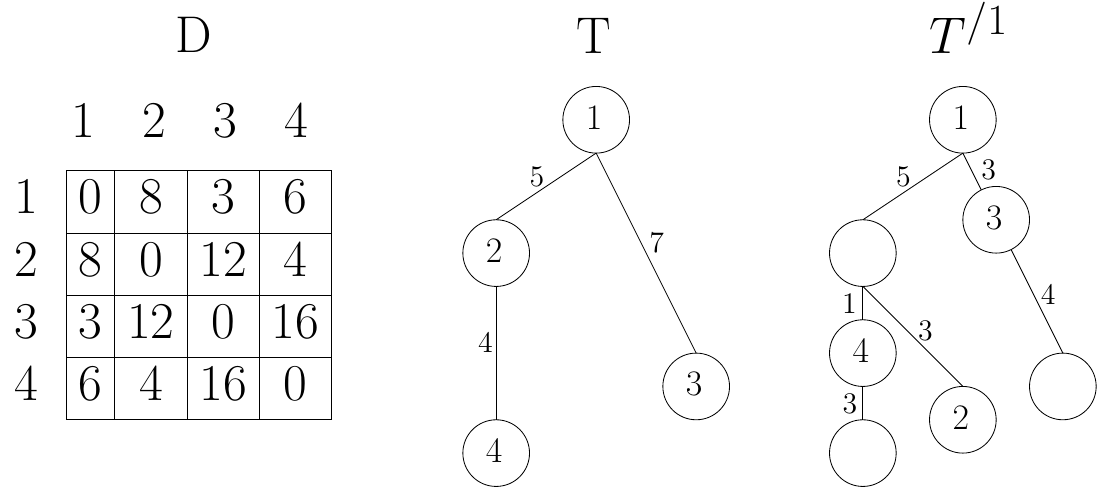}
    \caption{$T$ is not $\alpha$-restricted, for any $\alpha \in \{1,2,3,4\}$. By modifying $T$ we get $T^{/1}$ which is $1$-restricted. Nodes $3$ and $4$ move towards $1$, while $2$ moves away from $1$ (by creating a new leaf). Notice that some nodes of $T$ may be irrelevant for $T^{/1}$; however we do not need to explicitly delete them.}
    \label{fig:aRestricted}
\end{figure}

Intuitively, for any element $u$, if $T(\alpha,u) \ne D(\alpha,u)$, we move $u$ either closer to or further from $\alpha$.
More specifically, if $T(\alpha,u)>D(\alpha,u)$, then:
\begin{itemize}
    \item if there exists a node in the path from $u$ to $\alpha$ at distance $D(\alpha,u)$ from $\alpha$, we associate $u$ with this node,
    \item else there exists an edge in the path from $u$ to $\alpha$ such that one of its endpoints is at distance less than $D(\alpha, u)$ from $\alpha$, and the other endpoint is at distance larger than $D(\alpha,u)$ from $\alpha$.
    In this case we subdivide this edge in order to introduce a new node at distance exactly $D(\alpha,u)$ from $\alpha$.
    Then we associate $u$ with this new node.
\end{itemize}
Else if $D(\alpha,u) > T(\alpha,u)$ we create a new leaf under the node previously associated with $u$; the length of the edge connecting them is $D(\alpha,u)-T(\alpha,u)$.
Then we associate $u$ with the newly created leaf, instead of its parent.

The proof strategy for our main result is the following:

\begin{itemize}
    \item In order to approximate the optimal tree, it suffices to approximate the optimal $\alpha$-restricted tree, for some $\alpha\in V$. We prove this in Section~\ref{sec:tree}.
    \item In order to approximate the optimal $\alpha$-restricted tree, it suffices to approximate $L_0$ Fitting Constrained Ultrametrics. The proof directly follows from \cite{agarwala}; we include it in Appendix~\ref{app} for completeness.
    \item In order to approximate $L_0$ Fitting Constrained Ultrametrics, it suffices to approximate $L_0$ Fitting Ultrametrics. We prove this in Section~\ref{sec:ultrametric}.
\end{itemize}

\subsection{From tree metrics to restricted tree metrics} \label{sec:tree}

The main point where the reduction from \cite{agarwala} breaks for $L_0$ is the connection between optimal tree metrics and optimal $\alpha$-restricted tree metrics. The reason is that an averaging argument used to prove the result uses the convexity of $L_p$-norms, $p\ge 1$. In fact, this argument shows the existence of an element $\alpha$ such that the optimal $\alpha$ restricted tree is very close to the optimal tree. Then one can try all elements $\alpha$ and pick the best of them.

We show that even though $L_0$ is not convex, we can still select $\alpha$ in the same way as in \cite{agarwala}. In particular, we show that for any tree $T$, there exists an element $\alpha$ such that its $\alpha$-restricted tree metric $T^{/\alpha}$ does not increase the cost by more than a factor $3$.

\begin{lemma} \label{lem:modifiedReduction}
Let $D$ be a distance matrix and $T$ be a tree. Then there exists an element $\alpha$ such that for the $\alpha$-restricted tree $T^{/a}$ of $T$ it holds that $\|D-T^{/a}\|_0\le 3 \|D-T\|_0$.
\end{lemma}
\begin{proof}
We simply let $\alpha$ be the element that minimizes disagreements, that is 
\[\alpha = argmin_{u\in V} \|D(u)-T(u)\|_0\]
where $D(u)$ is the distance matrix $D$ restricted on the pairs containing $u$ (similarly for $T(u)$).
We have that $\|D-T\|_0 = \frac12 \sum_{u\in V} \|D(u)-T(u)\|_0$, as the sum in the right hand side double counts every pair $u,v$ with $D(u,v) \ne T(u,v)$.
By definition of $\alpha$, every term in the sum is lower bounded by $\|D(\alpha)-T(\alpha)\|_0$, meaning that 
\[\|D-T\|_0 \ge \frac{n}2 \|D(\alpha)-T(\alpha)\|_0\]

We say that an element $u$ is good if $D(\alpha,u)=T(\alpha,u)$, and bad otherwise; notice that by definition, the number of bad elements is exactly $\|D(\alpha)-T(\alpha)\|_0$.
As any element $u$ can have at most $n-1$ disagreements in $T^{/\alpha}$ (that is $\|D(u)-T^{/\alpha}(u)\|_0 \le n-1$), it follows that the number of pairs $u,v$ with at least one of $u,v$ being bad and $D(u,v) \ne T^{/\alpha}(u,v)$ is at most $\|D(\alpha)-T(\alpha)\|_0 \cdot (n-1)\le 2\|D-T\|_0$.

On the other hand, notice that if both $u,v$ are good, then by construction $T^{/\alpha}(u,v) = T(u,v)$. Therefore, if $D(u,v)\ne T^{/\alpha}(u,v)$, it also holds that $D(u,v)\ne T(u,v)$. The number of such pairs is upper bounded by $\|D-T\|_0$.
\end{proof}

Letting $T$ be an optimal solution for $L_0$ Fitting Tree Metrics establishes that it suffices to approximate the optimal $\alpha$-restricted tree, only increasing the approximation by a factor $3$.

Furthermore, we can directly use the techniques from \cite{agarwala} to show that any approximation for the constrained ultrametrics problem can give the exact same approximation for the optimal $\alpha$-restricted tree. We include this proof in Appendix~\ref{app} for completeness, as it is itself very brief. However we do not include it in the main body of the paper, as it has been extensively used in the literature.

Notice that these results already show that a $\rho$ approximation for $L_0$ Fitting Constrained Ultrametrics translates to a $3\rho$ approximation for $L_0$ Fitting Tree Metrics. This is an extension of the result of \cite{agarwala} for the $L_0$ case.

\subsection{From constrained ultrametrics to ultrametrics} \label{sec:ultrametric}
In this section we show that it is sufficient to approximate $L_0$ Fitting Ultrametrics, which is more natural than $L_0$ Fitting Constrained Ultrametrics. The technique used follows the one used in \cite{debarati} for $L_p$, $p\in \{1,2,\ldots\} \cup \{\infty\}$. However, in the case of $L_0$ we can simplify.

The high-level view of the technique is the following: Instead of trying to find a constrained ultrametric close to a distance matrix $D$, we rather squeeze $D$ itself to obey the constraints. Let $S_D$ be the resulting distance matrix. Then we find an (unconstrained) ultrametric $U$ close to this matrix $S_D$; due to the extra structure we imposed on $S_D$, we can only improve $U$ if we again squeeze it to obey the constraints. The resulting ultrametric $S_U$ is a constrained ultrametric, but all we needed to obtain it was a black-box algorithm for general ultrametrics.

We now define the squeezing process more formally.
In the $L_0$ Fitting Constrained Ultrametrics problem, for every element $u$ we are given a value $l_u$ which we call $u$'s lower-bound.
Furthermore we are given an upper bound $h$.

A constrained ultrametric $U$ shall satisfy that 
\begin{align*}
    h\geq U(u,v)\geq \max\{l_u,l_v\} \quad\quad \forall \{u,v\}\in \binom{V}{2}
\end{align*}

Given a distance matrix $A$, we define the \emph{squeezed} $A$ as the distance matrix $S_A$ for which $S_A(u,v)=\min\{ h,\max\{D(u,v),l_u,l_v\}\}$, for all $\{u,v\}\in \binom{V}{2}$. Intuitively, $S_A$ is obtained by squeezing $A$ to fit the constraints.

We use the well-known characterization of ultrametrics, that $U$ is an ultrametric iff $\forall \{u,v,w\} \in \binom{V}{3}: U(u,v) \le \max\{U(u,w),U(v,w)\}$.

\begin{lemma} \label{lem:ultrametricToConstrainedUltrametric}
A factor $\rho \ge 1$ approximation for $L_0$ Fitting Ultrametrics implies a factor $2 \rho$ approximation for $L_0$ Fitting Constrained Ultrametrics.
\end{lemma}
\begin{proof}
Our approach starts with creating $S_D$, the squeezed $D$. Notice that if $U'$ is a constrained ultrametric, then $\|U'-S_D\|_0 \le \|U'-D\|_0$. This follows because for any $u,v$ it holds that $\max\{l_u,l_v\} \le U'(u,v) \le h$, due to $U'$ being a constrained ultrametric. Therefore $D(u,v)=U'(u,v)$ only if $\max\{l_u,l_v\} \le U'(u,v) \le h$. But in this case $S_D(u,v)=D(u,v)=U'(u,v)$.

Similarly, suppose we have an ultrametric $U$, and we create the squeezed $S_U$.

With the exact same reasoning, we have 
\begin{equation}
    \|S_D-S_U\|_0 \le \|S_D-U\|_0 \label{eq:dPrimeAndUPrime}
\end{equation}

Our solution to $L_0$ Fitting Constrained Ultrametrics is to first create $S_D$ by squeezing $D$, then obtain ultrametric $U$ by a $\rho$ approximation to $L_0$ Fitting Ultrametrics on $S_D$, and finally obtain $S_U$ by squeezing $U$.

Let $OPT_{D,C}$ be the closest constrained ultrametric to $D$, and $OPT_{S_D}$ be the closest ultrametric to $S_D$. It suffices to show that $\|D-S_U\|_0 \le 2\rho \|D-OPT_{D,C}\|_0$ and that $S_U$ is indeed an ultrametric.

By definition of $S_D$, and since $OPT_{D,C}$ is constrained, for any two elements $u,v$ it holds that
\[\min\{D(u,v),OPT_{D,C}(u,v)\}\le S_D(u,v) \le \max\{D(u,v),OPT_{D,C}(u,v)\}\]

The proof follows by a straightforward case analysis of the $3$ cases $D(u,v)\le \max\{l_u,l_v\}$, $\max\{l_u,l_v\} < D(u,v) \le h$, $h < D(u,v)$. Therefore:
\begin{itemize}
    \item either $D(u,v)=OPT_{D,C}(u,v)$, in which case 
    \[|D(u,v)-OPT_{D,C}(u,v)|^0 = 0 = |D(u,v)-S_D(u,v)|^0 + |S_D(u,v)-OPT_{D,C}(u,v)|^0\]    
    \item or $D(u,v)\ne OPT_{D,C}(u,v)$, in which case 
    \[|D(u,v)-OPT_{D,C}(u,v)|^0 = 1,  |D(u,v)-S_D(u,v)|^0 + |S_D(u,v)-OPT_{D,C}(u,v)|^0 \le 2\]
\end{itemize}

We conclude that

\begin{equation}
    |D(u,v)-S_D(u,v)|^0 + |S_D(u,v)-OPT_{D,C}(u,v)|^0 \le 2|D(u,v)-OPT_{D,C}(u,v)|^0 \label{eq:triangleEquality}
\end{equation}

We now have 

\begin{align*}
\|D-S_U\|_0 &\le \|D-S_D\|_0 + \|S_D - S_U\|_0 & \text{(triangle inequality)}\\
&\le \|D-S_D\|_0+\|S_D-U\|_0 &\text{(\ref{eq:dPrimeAndUPrime})}\\
&\le \|D-S_D\|_0+\rho\|S_D-OPT_{D,C}\|_0 &\text{(definition of $U$)}\\
\end{align*}
As $\rho\ge 1$, the latter is upper bounded by
\begin{align*}
& \rho \sum_{\{u,v\}\in \binom{V}{2}}( |D(u,v)-S_D(u,v)|^0+|S_D(u,v)-OPT_{D,C}(u,v)|^0) \\
& \le 2\rho \sum_{\{u,v\}\in \binom{V}{2}}( |D(u,v)-OPT_{D,C}(u,v)|^0) &\text{(\ref{eq:triangleEquality})}\\
&= 2\rho \|D-OPT_{D,C}\|_0
\end{align*}

Finally, we need to prove that $S_U$ inherits that it is an ultrametric. This is clear if we proceed in rounds; each round we construct a new ultrametric, and the last one will coincide with $S_U$.

More formally, let $U_0=U$. In the first $|V|$ rounds, we take out a different $u'\in V$ at a time, and let

\[{U_r(u',v)=\max\{U_{r-1}(u',v),l_{u'}\}\quad\quad \forall v\ne u'}\]

Suppose $r>0$ is the first round where $U_r$ is not an ultrametric. Then there exists a triplet $\{u,v,w\}$ such that
$U_r(u,v) > \max\{U_{r}(u,w), U_{r}(v,w)\}$.
As we only increase distances, this may only happen if $U_r(u,v)>U_{r-1}(u,v)$. But this means that 
at round $r$ we picked either $u$ or $v$ (w.l.o.g. assume it was $u$) and set $U_r(u,v)=
l_u$. However, this would also give $U_r(u,w)\ge l_u = U_r(u,v)$, contradicting $U_r(u,v) > \max\{U_{r}(u,w), U_{r}(v,w)\}$.

Finally, for $S_U$ we simply have
\[S_U(u,v)=\min\{h,U_{|V|}(u,v)\}\]
Suppose there exists a triplet $\{u,v,w\}$ that now violates the ultrametric property, then it holds that

\[S_U(u,v) > \max\{S_U(u,w),S_U(v,w)\}\]

As we did not increase any distance of $U_{|V|}$, this means that $S_U(u,w)<U_{|V|}(u,w)$ or $S_U(v,w)<U_{|V|}(v,w)$; but distances can only reduce to $h$ which is an upper bound on $S_U(u,v)$ by construction.
\end{proof}

To prove our main theorem we use the following result from \cite{agarwala}. For completeness, we provide its proof in Appendix~\ref{app}.

\begin{restatable}{lemma}{implicitAgarwala}[Implicit in the proof of Lemma~3.5 of \cite{agarwala}] \label{lem:agarwala}
Let $D$ be a distance matrix, $\alpha$ be an element, and $T$ be an $\alpha$-restricted tree metric minimizing $\|T-D\|_0$. Assuming a $\gamma$-approximation to $L_0$ Fitting Constrained Ultrametrics, we can find an $\alpha$-restricted tree $T'$ such that $\|T'-D\|_0 \le \gamma \|T-D\|_0$.
\end{restatable}

We are now ready to prove our main theorem.

\thmMain*
\begin{proof}
We iterate over all $u\in V$ nodes.
In every iteration we use Lemma~\ref{lem:agarwala} along with a $2\rho$ approximation for $L_0$ Fitting Constrained Ultrametrics (obtained by Lemma~\ref{lem:ultrametricToConstrainedUltrametric}) to obtain a tree $T_u$ with $\|T_u-D\|_0 \le 2\rho \|T_u'-D\|_0$, where $T_u'$ is the optimal $u$-restricted tree metric.
Out of all the trees $T_u$ that we obtain, we output $T$, the one that minimizes $\|T_u-D\|_0$.

Let $T_{OPT}$ be an optimal tree metric.
By Lemma~\ref{lem:modifiedReduction} there exists an element $\alpha$ such that for the $\alpha$-restricted tree $T_{OPT}^{/\alpha}$ of $T_{OPT}$ it holds that $\|T_{OPT}^{/\alpha}-D\|_0 \le 3\|T_{OPT}-D\|_0$. 
Therefore there exists an element $\alpha$ for which we have $\|T'_\alpha-D\|_0 \le 3 \|T_{OPT}-D\|_0$.

It now holds that $\|T-D\|_0 \le \|T_\alpha-D\|_0 \le 2\rho \|T'_\alpha-D\|_0 \le 6\rho \|T_{OPT}-D\|_0$.
\end{proof}

\begin{corollary}
There exists a polynomial time $O(1)$ approximation for $L_0$ Fitting Tree Metrics.
\end{corollary}
\begin{proof}
Follows immediately, by using the polynomial time $O(1)$ approximation for $L_0$ Fitting Ultrametrics from \cite{vincent}.
\end{proof}

\section{APX-Hardness} \label{sec:apx}
In this section we show that $L_0$ Fitting Tree Metrics is APX-Hard.
Assuming, for the sake of contradiction, that it is not the case, we show how to approximate Correlation Clustering (an APX-Hard problem \cite{apxCorrClust}) within any constant factor.

This is a standard reduction used for $L_p$ Fitting Tree Metrics, $p\ge 1$.
It is however further simplified for $L_0$.
That is because once we decide to move a node, our cost does not depend on the distance we moved it.

In Correlation Clustering, we are given an unweighted undirected graph $G$, and the goal is to output a partition of the vertices\footnote{To avoid confusion, we use the term vertices when we refer to Correlation Clustering, and nodes when we refer to a tree.} (clustering) such that we minimize the number of pairs of vertices connected by an edge in $G$ that are in different parts of the partition (clusters) plus the number of pairs of vertices not connected by an edge in $G$ that are in the same part of the partition.

The idea behind the reduction is the following: for every pair of vertices connected by an edge in $G=(V,E)$, we set their distance to $0$, and for every pair of vertices not connected by an edge, we set their distance to something larger ($2$ in our case).
Then we solve $L_0$ Fitting Tree Metrics.
If the output tree has a good structure (every pair of nodes associated with elements of $V$ has equal distance), it would directly correspond to a clustering.
Namely, each node associated with elements of $V$ corresponds to a cluster that contains all elements associated with it (may be more than one).
Even though we can guarantee that an optimal solution has this structure, our approximation may not.

To fix this, we introduce many more elements at distance $0$ from each other, and at distance $1$ from every element in $V$.
This has the effect of maintaining the structure of an optimal solution, while incuring a big error in every solution that does not have the desired structure.

We now formally prove our result.

\begin{theorem}
$L_0$ Fitting Tree Metrics is APX-Hard.
\end{theorem}
\begin{proof}
Let $G=(V,E)$ be the input to Correlation Clustering, and $\epsilon>0$ be a sufficiently small constant.
For the sake of contradiction we assume that $L_0$ Fitting Tree Metrics can be approximated within an $1+\epsilon$ factor.
Then we show that using this approximation, we can approximate Correlation Clustering within the same factor.

Let $V'$ be a set, disjoint from $V$, of size $|V'|= 2 \binom{|V|}{2}$.
For any two elements $u',v'\in V'$ we have $D(u',v')=0$.
For $u,v\in V$ we have that $D(u,v)=0$ if $\{u,v\} \in E$, and $2$ otherwise.
Finally, for $u\in V, u'\in V'$ we have $D(u,u')=1$. 
Let $T$ be the tree output by $L_0$ Fitting Tree Metrics on $D$.

An upper bound for the optimal value is $\binom{|V|}{2}$.
To see this, create the tree $T'$ consisting of two nodes $u_{T'}, v_{T'}$ at distance $1$.
All elements of $V'$ are associated with $u_{T'}$, and all elements of $V$ are associated with $v_{T'}$.
As the only disagreements between $T'$ and $D$ are pairs of elements of $V$, the upper bound follows.

Furthermore, any solution that does not have all elements of $V'$ associated with the same node in $T$ has cost at least $|V'|-1$.
Therefore, for sufficiently small $\epsilon$, $T$ must have all elements of $V'$ associated with the same node $v'$.
Similarly, all elements of $V$ must be associated with nodes at distance $1$ from $v'$.
In particular, this corresponds to an ultrametric, where the root node is $v'$, and all leaves are at depth $1$.
In what follows we consider this tree rooted at $v'$.

Finally, if any non-root node of $T$ is at distance less than $1$ from $v'$, we remove it by connecting all its children with its parent node. Notice that this does not increase the number of disagreements, because the distance between elements of $V$ is either $0$ or $2$.
After we can no longer remove any node, we are left with a tree $T$ with root $v'$, and children $v_1, \ldots, v_\ell$ at distance $1$ from $v'$ (thus distance $2$ from each other).
Each $v_i$ is associated with some elements from $V$.
Our solution to Correlation Clustering is the partition of $V$ induced by $v_1, \ldots, v_\ell$.

By construction of $D$, the cost of this Correlation Clustering solution is exactly equal to $\|T-D\|_0$.
Furthermore, if $C_1, \ldots, C_{\ell'}$ is the optimal solution to Correlation Clustering, we can create a tree with a root $v'$ and children $v_1, \ldots, v_{\ell'}$ such that all elements of $V'$ are associated with $v'$ and all elements of $C_i$ are associated with $v_i$.
This means that the optimal Correlation Clustering cost is an upper bound to the optimal $L_0$ Fitting Tree Metrics.

We conclude that we found an $1+\epsilon$ approximation for Correlation Clustering, contradicting its APX-Hardness.
\end{proof}

We note that the proof assumes some distances to be $0$. If we want distances to be strictly positive, we can select a sufficiently small constant $\delta$ instead of $0$. Then, we replace nodes that are associated with multiple elements with stars whose leaves all have distance $\delta$ to each other.

\bibliography{bibli}

\appendix
\section{From restricted tree metrics to constrained ultrametrics} 
\label{app}
In this section, we show how to use a $\gamma$ approximation to $L_0$ Fitting Constrained Ultrametircs in order to find a $\gamma$ approximation to the optimal $\alpha$-restricted tree metric. We include this section for completeness, but the results follow directly from \cite{agarwala}.

We say that a distance matrix $Q$ is a centroid quasimetric if $\exists l_1, l_2, \ldots, l_n$ such that $Q(u,v)=\frac{l_u+l_v}{2}$. Notice that $l_u$ may even be negative, for any $u$.

Let $D$ be a distance matrix, $\alpha$ be a distinguished element, $m_\alpha = max_{u\in V\setminus \{\alpha\}} D(a,u)$, and $C^{\alpha}$ be a distance matrix such that $C^{\alpha}(u,v) = 2m_\alpha - D(a,u) - D(a,v)$. Notice that $C^{\alpha}$ is a centroid quasimetric by setting $l_u=2m_\alpha - 2D(a,u)$.
We say that $C^\alpha$ is the $(D,\alpha)$ centroid metric. We have the following two results.

\begin{lemma}[\cite{helpLemmas}, Th.3.2]\label{lem:helpLemma1}
For any element $\alpha$, $T$ is a tree metric if and only if $T + C^\alpha$ is an ultrametric.
\end{lemma}
\begin{lemma}[\cite{helpLemmas}, Cor.3.3]\label{lem:helpLemma2}
Given a tree metric $U$ and a centroid quasimetric $Q$, it holds that $U+Q$ is a tree metric if and only if $U+Q$ satisfies the triangle inequality.
\end{lemma}

We are now ready to prove the claim.

\implicitAgarwala*
\begin{proof}
Let $m_\alpha = max_{u\in V\setminus \{\alpha\}} D(a,u)$. Let $C^{\alpha}$ be a distance matrix such that $C^{\alpha}(u,v) = 2m_\alpha - D(a,u) - D(a,v)$. Furthermore we define the $L_0$ Fitting Constrained Ultrametrics problem with distance matrix $D'=D+C^{\alpha}, h=2m_\alpha, l_u=2m_a-2D(a,u)$ for all $u$. 

We now argue that there is an one to one correspondence between constrained ultrametrics and $\alpha$-restricted tree metrics, such that if constrained ultrametric $U$ correponds to tree $T'$, then $\|U-D'\|_0 = \|T'-D\|_0$. This directly proves the lemma's claim.

We start with the correspondence. Given a constrained ultrametric $U$, its corresponding $\alpha$-restricted tree is $U-C^\alpha$; similarly given an $\alpha$-restricted tree $T'$, its corresponding constrained ultrametric is $T'+C^{\alpha}$. Assuming this correspondence is valid, notice that $\|U-D'\|_0 = \|(T'+C^\alpha) - (D+C^\alpha)\|_0 = \|T'-D\|_0$.

We now prove that if $U$ is a constrained ultrametric, then $T'=U-C^\alpha$ is an $\alpha$-restricted tree metric. Notice that for any element $u\ne \alpha$, we have $U(\alpha,u)=h=2m_\alpha$, by $U$ being constrained. Therefore $T'(\alpha,u) = U(\alpha,u)-C^\alpha(\alpha,u) = 2m_\alpha - (2m_\alpha-D(\alpha, \alpha) - D(\alpha,u)) = D(\alpha,u)$, meaning that $T'$ is $\alpha$-restricted. To prove it is a tree metric, we need to prove that it satisfies triangle inequality, by Lemma~\ref{lem:helpLemma2}. Let $u,v,w$ be distinct elements.

\begin{align*}
    T'(u,v) &\le T'(u,w) + T'(v,w) &\iff\\
    U(u,v)-C^\alpha(u,v) &\le U(u,w)-C^\alpha(u,w) + U(v,w)-C^\alpha(v,w) &\iff\\
    U(u,v) &\le U(u,w) + U(v,w) - l_w &\iff\\
    U(u,v) &\le \max\{U(u,w), U(v,w)\} + \min\{U(u,w), U(v,w)\} - l_w
\end{align*}
But $U(u,v) \le \max\{U(u,w), U(v,w)\}$ since $U$ is an ultrametric, and $\min\{U(u,w), U(v,w)\} \ge l_w$ since $U$ is constrained. Therefore $T$ satisfies triangle inequality, and by Lemma~\ref{lem:helpLemma2} it is a tree metric.

Finally we prove that if $T'$ is an $\alpha$-restricted tree metric, then $U=T'+C^\alpha$ is a constrained ultrametric.
As $T'$ is $\alpha$-restricted, the $(D,\alpha)$ centroid metric is the same as the $(T',\alpha)$ centroid metric. 
Thus, by Lemma~\ref{lem:helpLemma1} $U$ is an ultrametric.
To show that it is constrained, notice that $l_\alpha = 2m_\alpha - 2D(\alpha, \alpha) = 2m_\alpha = h$. This also implies that for all $u$ we have $U(\alpha, u) = 2m_\alpha$.

Furthermore, $T'$ is a (tree) metric and therefore it satisfies triangle inequality. Let $u,v\ne \alpha$ be distinct elements.

\begin{align*}
    T'(\alpha, v) &\le T'(u,v) + T'(\alpha,u) &\implies \\
    U(\alpha, v) - C^\alpha(\alpha,v) &\le U(u, v) - C^\alpha(u,v) + U(\alpha, u) - C^\alpha(\alpha,u) &\implies \\
    U(\alpha, v) &\le U(u,v) + U(\alpha, u) - l_u &\implies\\
    2m_\alpha &\le U(u,v) + 2m_\alpha - l_u &\implies\\
    U(u,v) &\ge l_u
\end{align*}
By symmetry also $U(u,v) \ge l_v$, meaning that $U(u,v) \ge \max\{l_u,l_v\}$.

Finally
\begin{align*}
U(u,v) &= T'(u,v) + 2m_\alpha - D(\alpha,u) - D(\alpha,v)\\
&\le T'(\alpha, u) + T'(\alpha,v) + 2m_\alpha - D(\alpha,u) - D(\alpha,v) \\
&= D(\alpha, u) + D(\alpha,v) + 2m_\alpha - D(\alpha,u) - D(\alpha,v) \\
&= 2m_\alpha
\end{align*}

\end{proof}

\end{document}